%% file: CDC26WorkingVersion.tex
\documentclass[letterpaper, 10 pt, conference]{ieeeconf}  

\IEEEoverridecommandlockouts                              
    
\makeatletter
\g@addto@macro\normalsize{
  \setlength{\abovedisplayskip}{6pt}
  \setlength{\belowdisplayskip}{4pt}
  \setlength{\abovedisplayshortskip}{0pt}
  \setlength{\belowdisplayshortskip}{3pt}
}
\makeatother

\usepackage{titlesec}
\titlespacing*{\section}{0pt}{1ex}{0.5ex}
\titlespacing*{\subsection}{0pt}{0.8ex}{0.4ex}

\usepackage[skip=4pt]{caption} 
\usepackage{graphicx} 
\usepackage{amsmath,amsfonts,amssymb}
\usepackage{booktabs}
\usepackage{multicol}
\usepackage{multirow}
\usepackage{xcolor}
\usepackage[commentmarkup=uwave,]{changes}
\definechangesauthor[color=purple]{RW}
\definechangesauthor[color=violet]{AC}
\definechangesauthor[color=blue]{ES}
\definechangesauthor[color=red!70!black]{IM}

\newcommand{\R}{\mathbb{R}}

\newtheorem{lemma}{Lemma}
\newtheorem{theorem}{Theorem}
\newtheorem{corollary}{Corollary}
\newtheorem{proposition}{Proposition}
\newtheorem{definition}{Definition}

\newtheorem{claim}{Claim}

\input{MathDefs}

\newcommand{\cdcsubmission}[1]{#1}
\newcommand{\extendedversion}[1]{#1}
\renewcommand{\cdcsubmission}[1]{}

\title{\LARGE \bf
Remarks on Lipschitz-Minimal Interpolation:\\Generalization Bounds and Neural Network Implementation
}

\author{Arthur C. B. de Oliveira$^{1}$, Ruigang Wang$^{2}$,  Ian R. Manchester$^{2}$,  Eduardo D. Sontag$^{1}$
\thanks{*This research was partially supported by ONR grant N00014-21-1-2431, AFOSR grant FA9550-2110289, and the Australian Research Council through projects DP230101014 and IH210100030.}
\thanks{$^{1}$\mbox{Depts of Electrical \& Computer Engineering and BioEngineering,} \mbox{Northeastern University, Boston, Massachusetts, USA.}
        {\tt\small castello.a@northeastern.edu, e.sontag@northeastern.edu}}%
\thanks{$^{2}$\mbox{Australian Centre for Robotics and School of Aerospace, Mechanical} \mbox{and Mechatronic Engineering, University of Sydney, Australia.}
        {\tt\small ruigang.wang@sydney.edu.au, ian.manchester@sydney.edu.au}}%
}

\begin{document}

\maketitle
\thispagestyle{empty}
\pagestyle{empty}

\begin{abstract}
    This note establishes a theoretical framework for finding (potentially overparameterized) approximations of a function on a compact set with a-priori bounds for the generalization error. The approximation method considered is to choose, among all functions that (approximately) interpolate a given data set, one with a minimal Lipschitz constant. The paper establishes rigorous generalization bounds over practically relevant classes of approximators, including deep neural networks. It also presents a neural network implementation based on Lipschitz-bounded network layers and an augmented Lagrangian method. The results are illustrated for a problem of learning the dynamics of an input-to-state stable system with certified bounds on simulation error.
\end{abstract}

\section{Introduction}

Modern machine learning practice is characterized by the training of large models, such as deep neural networks, via simple first-order optimization algorithms. A remarkable fact is that highly over-parameterized models -- i.e. models with many more parameters than are required for perfect interpolation of the training data -- frequently exhibit impressive generalization performance, seemingly contradicting the classical bias-variance trade-off and leading to a so-called ``double-descent'' phenomenon \cite{belkin2019reconciling, nakkiran2021deep}.

 One explanation for this surprising observation is a form of Occam's razor: to generalize well one should choose, among all interpolating models, the \textit{smoothest} model in some sense. In \cite{belkin2021fit} smoothness was considered in relation to kernel methods via Sobolev-like norms and it was argued that over-parameterization generally enriches the model set thus admitting smoother models, and some intrinsic properties of the model class and/or the learning algorithm  lead to an inductive bias towards smoother models.

While large models often generalize well in a statistical sense, it has been observed that deep neural networks can be very sensitive to \textit{adversarial} (worst-case) perturbations \cite{szegedy2013intriguing}, i.e. they can have very large Lipschitz constants. This lack of robustness raises questions about the suitability of NNs for use in safety-critical settings such as control systems \cite{manchester2026neural} and has led to a large body of research on novel training methods and models with certified Lipschitz bounds to improve robustness, see e.g. \cite{miyato2018spectral, virmaux2018lipschitz, trockman2021orthogonalizing, wang2023direct, araujo2023unified, manchester2026neural}.

Motivated by these observations, in this paper we consider the problem of learning an interpolating model which is \textit{explicitly} optimized to be the smoothest in the sense of \textit{minimal Lipschitz constant}. This problem connects to classical results on the extension of functions: in 1934 Kirszbraun proved a function defined on a subset of a Hilbert space (e.g. a finite data set) can be extended to the whole Hilbert space without increasing its Lipschitz constant \cite{kirszbraun1934zusammenziehende}. While the original proof was non-constructive, more constructive formulations have been found \cite{azagra2021kirszbraun} and recent work has developed an efficient online algorithm for computing an extension \cite{zaichyk2024efficient}.

In this paper, we consider learning an interpolator from a parameterized set of  models  such as deep neural networks.
In such a case, recent theoretical work suggests that significant overparameterization is \textit{necessary}: interpolating with a small Lipschits constant requires many more parameters than interpolation without any smoothness condition \cite{bubeck2021universal}.


The first contribution of this paper is a set of theoretical results on Lipschitz-minimal interpolators, and approximate versions thereof. This work shows that such models have generalization bounds that are linear in model approximation error and the distance between test points and the training data, showing directly how richer models and more data lead to improved generalization.

Our main result, Theorem~\ref{thm:HighProbpUniDist}, shows that for $N$ i.i.d.\ samples $\{x_i,g(x_i)\}_{i=1}^N$ drawn uniformly from $[0,1]^n$, where $g$ is Lipschitz, the generalization error decays at rate $(\log(N)/ N)^{1/n}$ with high probability, provided that the hypothesis class $\FSet$ is dense in the $C^{0,1}$ topology (equivalently, under the norm $\|f\|_\infty + \|f'\|_\infty$, where $f'$ denotes the Jacobian of $f$), as is the case for deep neural networks.
Specifically, there exists at least one $f \in \FSet$ with sufficiently small Lipschitz constant that nearly interpolates the data, and any such near-interpolant satisfies a uniform loss bound
\[
\Loss(f,g) = O\left(\frac{\log N}{N}\right)^{1/n}
\]
with high probability.



A second contribution of the paper is to investigate neural network implementations of the proposed strategy. We develop specific algorithms for fitting with (near)-minimal Lipschitz constants, based on an augmented Lagrangian scheme and recently developed Lipschitz network layers \cite{tsuzuku2018lipschitz, trockman2021orthogonalizing, wang2023direct}. We illustrate the approach on a simple problem of fitting an input-to-state stable vector field, leading to theoretical guarantees on simulation error.

\section{Theoretical Setup and Preliminaries}
\subsection{Preliminary Definitions}
Let $\na,\ir,\re,\rep,\repp$ be the sets of natural, integer, real, nonnegative real, and positive real numbers, respectively. 

We denote by $\dist:\re[n]\times\re[\SDim]\rightarrow\rep$ the Euclidean distance between two points in $\re[\SDim]$. 
The distance between a point $x\in\re[\SDim]$ and a set $\xsp\subset\re[\SDim]$ is 
$\dist(x,\xsp):=\inf_{\bar x\in\xsp}\dist(x,\bar x)$. 
For a function $\func:\re[\SDim]\rightarrow\re$ and a set $\xsp\subset\re[\SDim]$, we call the image of $\func$ on $\xsp$
\[
\func[\xsp]:=\{f(x)\in\re~|~x\in\xsp\}.
\]

For a compact set $\SSpc\subset\re[\SDim]$, let $C(\SSpc)$ and $\LSpace\subset C(\SSpc)$ denote the sets of continuous, respectively Lipschitz, functions $\func:\SSpc\rightarrow\re$.
The \textit{Lipschitz semi-norm} of a function $\func\in C(\SSpc)$ is 
\begin{align*}
    \Lip(\func):=&\inf\{\gamma\in\rep\cup\infty~|~\\&\|f(x_1)-f(x_2)\|\leq\gamma\|x_1-x_2\|,~\forall x_1,x_2\in\SSpc\}.
\end{align*}
Thus elements of $\LSpace\subset C(\SSpc)$ are those with 
$\Lip(\func)<\infty$.
Finally, we denote by $C^1(\SSpc)$ the set of continuously differentiable (and hence Lipschitz if $\SSpc$ is compact) functions on $\SSpc$.


Let $\SSpc\subset\re[\SDim]$ be a compact set.
Given a function $\funcg:\SSpc\rightarrow\re$ and any $\er\geq0$, an \emph{$\er$-noisy dataset generated by $\funcg$} is any set of points
$$
\nDSet[\funcg]:=\{(x_1,y_1),\dots,(x_\NSet,y_{\NSet})\}\subset\SSpc\times \re
$$
such that
$$
\|\fgen(x_i)-y_i\|\leq\er,
$$
for all $i=1,2,\dots,\NSet$. 
In the special case $\er=0$, we write simply $\DSet[\funcg]$ and call this a (perfect)
\emph{dataset associated to} $\funcg$. Note that then
$$
    \DSet[\funcg]:=\{(x_1,\funcg(x_1)),\dots,(x_\NSet,\funcg(x_{\NSet}))\}\subset\SSpc\times \re.
$$
We call $\NSet$ the \textit{size} of the set, and
any such function $\funcg$ is called a \textit{generator} of the dataset, and denote by $\xsp:=\{x_1,\dots,x_{\NSet}\}$ and $\ysp:=\{y_1,\dots,y_{\NSet}\}$ the (arbitrarily) ordered sets of inputs and respective outputs in $\DSet$.


Given two functions $\fgen,\func:\SSpc\to\re$, we denote the sup error or ``$L_\infty$-loss" as
\begin{equation}
    \Loss(f,g) := \max_{x\in\SSpc}\|g(x)-f(x)\| =: \|g-f\|_{\infty,\SSpc}.
\end{equation}

\subsection{Initial problem statement}

Suppose that there is an unknown function $\fgen:\SSpc\to\re$, $\fgen\in \LSpace$ thought of as a data generator,
and we wish to find, based on available data, an approximating 
function $\func:\SSpc\to\re$, $\func\in \LSpace$.
If complete information were available,
the problem of finding $\fgen$ can be written as solving for
\begin{equation}
    \label{eq:origiopt}
    \argmin_{f\in \LSpace}\Loss(f,g)=\argmin_{f\in \LSpace}\max_{x\in\SSpc}\|g(x)-f(x)\|.
\end{equation}
Obviously, $f=\fgen$ is the unique solution of this optimization problem. Of course, in practice one can only sample and evaluate $\fgen$ on finitely many points from $\SSpc$. This results in only a dataset $\DSet$ if the process is noise free, or more generally $\nDSet$. Given such a dataset, and the goal of estimating $\fgen$, we define the empirical loss minimization problem as
\begin{equation}
    \label{eq:sampprob}
    \argmin_{f\in \LSpace}\loss(\nDSet,\func).
\end{equation}
where the training loss is given by
\begin{equation}
    \label{eq:losssample}
    \loss(\nDSet,\func):=\max_{(x_i,y_i)\in\nDSet}\|y_i-f(x_i)\|.
\end{equation}
Contrary to \eqref{eq:origiopt}, \eqref{eq:sampprob} does not have a unique solution in general. Moreover, there is no a priori guarantee that any solution of \eqref{eq:sampprob} will even be ``close'' to $\fgen$ outside of the training points. Despite that, we can still state the following elementary, but key result on the generalization error of a solution of \eqref{eq:origiopt}. 
\smallskip
%
%
\begin{lemma}

    \label{lem:UBOrigiAnyPt}
    Given $\fgen\in\LSpace$, let $\nDSet$ be a noisy dataset with noise bound $\er$. Furthermore, for any $\eps>0$ let $\fopt$ be any function in $\LSpace$ that satisfies $\loss(\nDSet,\fopt)\leq \eps$, with $\lc{\fgen}$ and $\lc{\fopt}$ being the Lipschitz constants of $\fgen$ and $\fopt$ respectively. Then for any $x\in\SSpc$
    \begin{equation}
        \|\fgen(x)-\fopt(x)\|\leq (\lc{\fopt}+\lc{\fgen}) \dist(x,\xsp)+\er+\varepsilon
    \end{equation}
%
%
%
 %
\end{lemma}
\begin{proof}
    For any $x\in\SSpc$ it holds that
    \begin{align*}
        \|\fgen(x)&-\fopt(x)\|\\&= \|\fgen(x)-\fgen(x_i)+\fgen(x_i)-y_i\\&~~~~~~+y_i-\fopt(x_i)+\fopt(x_i)-\fopt(x)\|\\
        &\leq \|\fgen(x)-\fgen(x_i)\|+\|\fgen(x_i)-y_i\|\\&~~~~~~+\|y_i-\fopt(x_i)\|+\|\fopt(x_i)-\fopt(x)\|\\
        &\leq \lc{\fgen}\|x-x_i\|+\er+\varepsilon +\lc{\fopt}\|x-x_i\|,
    \end{align*}
    where $i:=\argmin_{i=1,\dots, N}\|x-x_i\|$ implying $\|x-x_i\|=\dist(x,\xsp)$, and completing the proof.
\end{proof}

{To strengthen this result into a general guarante on $\SSpc$, we must first assume that the sample $\xsp$ is a ``good representation'' of $\SSpc$, a notion that the \emph{fill distance/covering radius} formalizes:}

\begin{definition}[Fill Distance/Covering Radius \cite{reznikov2016covering}]
    Let $\SSpc\subset\mathbb{R}^n$ be compact subset with non-empty interior, and $\xsp:=\{x_1,\dots,x_N\}\subset\SSpc$. Then the \emph{fill distance/covering radius} of $\xsp$ with respect to $\SSpc$ is given by
    \begin{equation}
        \label{eq:defcovrad}
        h(\xsp,\SSpc):=\sup_{x\in\SSpc}\dist(x,\xsp).
    \end{equation}
\end{definition}
\smallskip
With this definition, a ``good representation'' is formalizd as a sample set whose covering radius is small enough. Using this, we can provide the following Corollary of Lemma \ref{lem:UBOrigiAnyPt}:

\begin{corollary}
    \label{cor:UBOrigiForm}
    Given $\fgen\in \LSpace$, let $\nDSet(\fgen)$ be a noisy dataset with noise bound $\er$ and covering radius $h(\xsp,\SSpc)=h>0$. Furthermore, for any $\varepsilon>0$, let $\fopt$ be any function in $\LSpace$ that satisfies $\loss(\nDSet,\fopt)\leq \eps$, with $\lc{\fgen}$ and $\lc{\fopt}$ being the Lipschitz constants of $\fgen$ and $\fopt$ respectively. Then
    \begin{equation}
        \Loss(\fopt,\fgen)\leq (\lc{\fopt}+\lc{\fgen}) h+\er+\varepsilon
    \end{equation}
\end{corollary}

\begin{proof}
    This corollary follows immediately from Lemma \ref{lem:UBOrigiAnyPt} and the definition of covering radius \eqref{eq:defcovrad}.
\end{proof}

Although a finite upper-bound on the value of $\Loss(\fgen,\fopt)$, the result in Corollary \ref{cor:UBOrigiForm} can be arbitrarily large for any arbitrarily small $\varepsilon>0$. In other words, given $\SSpc$, $\DSet$, and any $\delta>0$, there always exists an $\fopt$ that solves \eqref{eq:sampprob} -- and therefore satisfies the conditions in the lemma with $\varepsilon=0$ -- and that satisfies $\Lip(\fopt)>\delta$. 

This is well known in the context of polynomial fitting, and can be illustrated as follows: Given a dataset $\DSet$ build a new dataset by adding a new point $(x_{N+1},y_{N+1})$ as follows: pick any $x_{N+1}\in\SSpc$, and let $\bar i:=\argmin_{i=1,\dots,N}\|x_i-x_{N+1}\|$; then pick $y_{N+1}$ such that $\|y_{N+1}-y_{\bar i}\|\geq\delta\|x_{\bar i}-x_{N+1}\|$. For this new dataset, pick a function $\func$ that perfectly interpolates it. This function also perfectly interpolates $\DSet$ and has Lipschitz constant $\Lip(\func)\geq \delta$, by construction.

This example illustrates the issues with this approach while also indicating a possible path for correction: choosing approximating functions with small Lipschitz constants. We next look at an alternative formulation to the problem that leverages this intuition.

\section{Lipschitz-Minimal Interpolation}

Including Lipschitz constraints is not a new idea in machine learning, however in this section we will investigate a priori theoretical guarantees for solutions of \textit{Lipschitz-minimal} approximations. Given a noisy dataset $\nDSet$ and an error measure $\varepsilon\geq\er$, consider the following problem
\begin{subequations}
    \label{eq:lipregprob}
    \begin{align}
        \underset{\func \in \LSpace}{\textrm{minimize}} \quad &\Lip[\func]\label{eq:lip-cost}\\
        \text{subject to}\quad&\loss(\nDSet,\func)\leq \varepsilon.\label{eq:lipregprob-constraint}
    \end{align}
\end{subequations}

%
Notice that the condition $\eps\geq\er$ is sufficient to guarantee feasibility of the problem, since $g\in L(\SSpc)$. For \eqref{eq:lipregprob} consider the following result:
\begin{theorem}
    \label{thm:ThrGuarantee1}
    For any $\fgen\in \LSpace$, let $\nDSet(\fgen)$ be a noisy dataset with noise bound $\er$ and covering radius $h(\xsp,\SSpc)=h>0$. Pick any $\varepsilon\geq \er$. Then: (i) there always exists a $\fopt\in L(\SSpc)$ that solves \eqref{eq:lipregprob}; and (ii) any $f^*$ that solves \eqref{eq:lipregprob} satisfies 
    \begin{equation}
        \Loss(\fgen,\fopt)\leq 2(\lc{\fgen} h+\varepsilon),
    \end{equation}
    with $\lc{\fgen}$ being the Lipschitz constant of $\fgen$.
\end{theorem}
\begin{proof}
    First, assume (i) is true to prove (ii). In that case, notice that $\fgen$ is a feasible solution of \eqref{eq:lipregprob}, which means that for any $\fopt$ that solves \eqref{eq:lipregprob}, $\Lip(\fopt)\leq\Lip(\fgen)$ necessarily. Then, from Corollary \ref{cor:UBOrigiForm} we have that
    \[
    \|\fgen(x)-\fopt(x)\|\;\leq\; (\lc{\fgen}+\lc{\fopt})h+\varepsilon+\er
    \;\leq \;2(\lc{\fgen}h+\varepsilon),
    \]
     completing the proof of (ii).

     To prove (i), first consider for any $m>0$ the set
     \[
     \mathcal{F}_m:=\{\func\in \LSpace~|~\loss(\nDSet,\func)\leq\varepsilon\, ~\text{ and } \lc{\func}\leq m\},
     \]
     which is a subset of the feasible set of \eqref{eq:lipregprob}. This set is equibounded (in the sup norm), because any $\func\in\mathcal{F}_m$ satisfies $\|\func\|_\infty\le mh(\xsp,\SSpc)+\beta$, where $\beta:=\max_{y_i\in\ysp}(\|y_i\|+\er)$.

     Next, let $m=\Lip(\tilde\func)$ for some $\tilde\func$ such that $\Loss(\nDSet,\tilde\func)\le\eps$ (for example, pick an interpolant polynomial) and consider the set $\mathcal{F}_m$.

     We will argue next that $\mathcal{F}_m$ is a compact set in $C(\SSpc)$. To see that, first notice that it is equicontinuous, since all functions in it satisfy the same Lipschitz bound with Lipschitz constant equal to $m$ -- simply pick for every $\tilde\varepsilon>0$ a $\delta:=\tilde \varepsilon/m$, which by the definition of Lipschitz will imply that if $\|x-y\|\leq \delta$ then $\|f(x)-f(y)\|\leq \tilde\varepsilon$ for all $f\in\mathcal{F}_m$. Also, notice that $\mathcal{F}_m$ is equibounded, as shown earlier. Finally, since $\mathcal{F}_m$ is equicontinuous and equibounded, by the Arzel\`a–Ascoli Theorem, it is a compact set.

     Next, for completeness, we will prove that the Lipschitz seminorm is lower semicontinuous with respect to the uniform norm -- i.e. we will prove that for any sequence $\{f_k\}$ of Lipschitz functions converging to a Lipschitz function $f$,
     $$\Lip(f)\leq \liminf_{k\to\infty}\Lip(f_k).$$
     To see this, consider any $x_1,x_2\in\SSpc$ with $x_1\neq x_2$. By the triangle inequality
     \begin{align*}
         \|f(x_1)-f(x_2)\|&\leq \|f(x_1)-f_k(x_1)\|\\&~~~~~+\|f_k(x_1)-f_k(x_2)\|\\&~~~~~+\|f(x_2)-f_k(x_2)\|\\
         &\leq \|f_k(x_1)-f_k(x_2)\|+2\|f-f_k\|_\infty.
     \end{align*}
     
     Dividing everything by $d(x_1,x_2)$ and taking the $\liminf_{k\to\infty}$ results in
     \begin{align*}
         \frac{\|f(x_1)-f(x_2)}{d(x_1,x_2)}\|&\leq \liminf_{k\to\infty}\frac{\|f_k(x_1)-f_k(x_2)\|}{d(x_1,x_2)}\\&\leq\liminf_{k\to\infty} \Lip[f_k]
     \end{align*}
     Finally, since the inequality above holds for all $x_1,x_2\in\SSpc$, in particular it holds for
     \begin{equation*}
         \Lip(f)=\sup_{x_1\neq x_2}\frac{\|f(x_1)-f(x_2)\|}{d(x_1,x_2)}\leq\liminf_{k\to\infty}\Lip[f_k],
     \end{equation*}
     proving lower semi-continuity of the Lipschitz seminorm. Finally, since $\mathcal{F}_m$ is compact and $\Lip(\cdot)$ is lower semi-continuous then it attains its minimum at some $\fopt\in\mathcal{F}_m$, completing the proof of (i) and of the theorem.
\end{proof}

Notice that the result from Theorem \ref{thm:ThrGuarantee1} provides an upper-bound on $\Loss(\fopt,\fgen)$ that is independent of $\fopt$. This is an important distinction because it proves that for problem \eqref{eq:lipregprob}, the solution $\fopt$ can be made arbitrarily close to the generator function $\fgen$ by improving the dataset in terms of its covering radius $h$ and the measurement noise noise $\er$. 

\subsection{Optimization over a subset of Lipschitz functions}

Theorem \ref{thm:ThrGuarantee1} showed how minimizing the Lipschitz constant of the chosen solution gives a priori guarantees for the generalization error of the solution of the problem. However, the problem formulation in \eqref{eq:lipregprob} requires us to minimize over the entire set of Lipschitz functions, which is not viable in practice. Instead, the optimization is often performed over a subset of functions (for example polynomials, neural networks, etc) that does not necessarily contain the function $\fgen$ that generated the data. 

Formally, consider a subset of smooth Lipschitz functions $\FSet(\SSpc)\subseteq C^1(\SSpc)$ that is dense in the $C^{0,1}(\SSpc)$ topology -- i.e. dense w.r.t. the metric $\|\func\|_{\FSet}:=\|\func\|_\infty+\|\func'\|_\infty$, where $\func'$ is the Jacobian of $\func$. Notice that this is not a restrictive assumption, since popular classes of functions used in optimization satisfy it -- such as feedforward neural networks \cite{hornik1991approximation} with $C^1$ activations.

Given this set of functions, let $\fopt$ be any solution of \eqref{eq:lipregprob}.
Then we have the following result:
\begin{theorem}
    \label{thm:ThrGuarantee2}
    For any $\rho>0$, $\fgen\in \LSpace$, and $\nDSet(\fgen)$, let $\FSet(\SSpc)\subset C^1(\SSpc)$ be a set of functions dense in the $C^{0,1}(\SSpc)$ topology, and $\varepsilon\geq\er$ be the tolerance of \eqref{eq:lipregprob}. Then: (i) there always exists an $\fmin\in\FSet$ that is both a feasible solution of \eqref{eq:lipregprob} and satisfies $\Lip[\fmin]\leq \lopt+\rho$; and (ii) for any such $\fmin$ it holds that
    \begin{equation}
        \label{eq:OptProbFSetBound}
        \Loss(\fmin,g)\leq (2\lc{\fgen}+\rho)h+2\varepsilon.
    \end{equation}
\end{theorem}
\begin{proof}
    First, we assume (i) is true and prove (ii). If (i) is true, then apply Corollary \ref{cor:UBOrigiForm} to get
    \[
    \|\fgen(x)-\fopt(x)\|\;\leq\; (\lc{\fgen}+\lc{\fopt})h+\varepsilon+\er
    \;\leq\; (2\lc{\fgen}+\rho)h+2\varepsilon,
    \]
    proving (ii). 
    
    Next we provide a proof sketch of (i). \cdcsubmission{In~\cite{lipschitz_arxiv}, we discuss an alternative proof, presenting it in full.} \extendedversion{In the appendix we present a full alternate proof of the claim.}
    First, consider the classical method of approximating $L$-Lipschitz functions on $\SSpc$ via convolution with standard mollifiers. Given an $L$-Lipschitz function $g$, this procedure yields a sequence $f_n\in C^\infty$ with $\Lip(f_n)\leq L$ and $\Lip(f_n)\to L$. This follows because if $g$ is Lipschitz in $\re[n]$, then $g\in W^{1,\infty}(\re[n])$ (the Sobolev space of essentially bounded functions with essentially bounded gradient and the topology induced by the $\|\cdot\|_{\FSet}$ norm defined earlier), and mollification preserves the $L^\infty$ bound on the gradient; see e.g. \cite{EvansGariepy2015}. Consequently, the mollified functions are Lipschitz with no larger Lipschitz constant. An explicit formulation of this approximation result, and a substantial generalization to Lipschitz functions on Riemannian manifolds is given in \cite{azagra2007smooth}.

    Through this method, we can find a smooth $f_n$ that is $L$-Lipschitz and such that $\|f-g\|_\infty\leq\eps/2$. By density of $\FSet$ in the $C^{0,1}$ topology, it follows that there exists an $f\in\FSet$ such that $\|f-f_n\|_{\FSet}\leq\tilde\rho$ where $\tilde\rho=\min(\rho,\epsilon/2)$. From here, we recover the statement by noticing that if $\|f-f_n\|_{\FSet}\leq\tilde\rho$ then: $\|f-f_n\|_\infty\leq\tilde\rho\Rightarrow\|g-f\|_\infty\leq\eps$; and $\|f'-f_n'\|_\infty\leq\tilde\rho\Rightarrow\Lip(f)\leq \lc{\fgen}+\rho$.
\end{proof}


This is an important result, because there is no guarantee that \eqref{eq:lipregprob} has a solution if we restrict the search space to $\FSet$. However, even in that case, Theorem \ref{thm:ThrGuarantee2} guarantees that one can get arbitrarily close to a solution, motivating numerical methods.

Furthermore, the resulting inequality \eqref{eq:OptProbFSetBound} provides an important insight for applications: both the measurement noise $\eps$ and the precision of the dataset $h$ affect the global precision of the solution candidate in a linear manner. Arbitrarily increasing the size of the dataset will provide limited gains if the measurement noise is not also addressed, and vice-versa.

\subsection{The problem under random sampling}

Often in practice, one cannot sample in a perfect grid, or even pick which points to sample from. As a consequence, the bound on the covering radius $h$ might not be known with certainty. In this context, the following corollary follows from Theorem \ref{thm:ThrGuarantee2}.

\begin{corollary}
    \label{cor:highprobbnd}
    Let $\FSet\subset C^1(\SSpc)$ be a set of functions dense in $C^{0,1}(\SSpc)$, and let $\nDSet$ be such that $\xsp$ is sampled i.i.d. from $\SSpc$ following some probability distribution $\mu$ over $\SSpc$. For any $\rho>0$ and $\fgen\in \LSpace$, let $\lc{\fgen}:=\Lip(g)$, and $\eps\geq\er$ be the tolerance of \eqref{eq:lipregprob}. Furthermore, let $h>0$ and $\delta\in [0,1)$ be such that $h(\xsp,\SSpc)\leq h$ with probability $(1-\delta)$. Then: (i) there exists an $\fmin\in\FSet$ that is both a feasible solution of \eqref{eq:lipregprob} and satisfies $\Lip(\fmin)\leq l^*+\rho$; and (ii) for any such $\fmin$, with probability $(1-\delta)$ it holds that
    \begin{equation}
        \Loss(\fmin,\fgen)\leq(2\lc{\fgen}+\rho)h+2\eps.
    \end{equation}
\end{corollary}

This Corollary follows directly from Theorem \ref{thm:ThrGuarantee2}, and it provides a direct way of applying the results in a random sampling scenario. For example, let $\SSpc$ be a hypercube in $\re[n]$ and let $\nDSet$ be drawn i.i.d from a uniform distribution on $\SSpc$. 
For this setup, we have the following result:


\begin{theorem}
    \label{thm:HighProbpUniDist}
    Fix any $n\in\nat$, and let $\SSpc=[0,1]^n$ be the unit hypercube in $\re[n]$. Furthermore, let $\FSet(\SSpc)\subset C^1(\SSpc)$ be dense in the $C^{0,1}$ topology. Let $\fgen\in\LSpace$, and $\xsp=\{x_1,\dots,x_N\}$ be i.i.d. uniformly sampled from $\SSpc$, with $\nDSet$ being the resulting dataset obtained from $\fgen$. Then there exist constants $k_1>0$, $k_2>0$ and $N_0\in\nat$ whose value depend only on the dimension $n$, such that for any $\delta\in(0,1)$, $\rho>0$, and $\eps>\er$, and all $N\geq N_0$: (i) there exists a feasible solution $\fmin$ of \eqref{eq:lipregprob} that satisfies $\lc{\fmin}\leq \lc{\fopt}+\rho$; and (ii) for any such solution the following property holds:
    \begin{equation}
        \label{eq:ErBndUniSmpl}
        \mathbb{P}\left[\Loss(\fmin,\fgen)\leq(2\lc{\fgen}+\rho)k_1\left(\frac{\log\left(k_2N/\delta\right)}{N}\right)^{1/n}+2\eps\right]\geq 1-\delta.
    \end{equation}
\end{theorem}



\begin{proof}
    Claim (i) follows immediatelly from Theorem \ref{thm:ThrGuarantee2}, as it does not depend on the sampling structure. To prove (ii), we will first apply Theorem 2.1 of \cite{reznikov2016covering} to the uniform distribution over the hypercube. Let $\mu$ be the normalized Lebesgue measure over $\SSpc$ (uniform probability measure), the theorem in question states that for any class-$\mathcal{K}$ $\Phi:\re\to\re[+]$ for which there exists a number $r_0$ such that $\mu(B_r(x))\geq\Phi(r)$ for all $x\in\SSpc$ and $r<r_0$ (where $B_r(x)$ is the closed ball of radius $r$ centered in $x$), then there exist positive constants $c_1,c_2,c_3,$ and $\alpha_0$ such that for any $\alpha>\alpha_0$ we have that
    \begin{equation*}
        \mathbb{P}\left[h(\xsp,\SSpc)\geq c_1\Phi^{-1}\left(\frac{\alpha\log(N)}{N}\right)\right]\leq c_2N^{1-c_3\alpha}.
    \end{equation*}
    We begin by showing that for the uniform distribution over a hypercube, $\Phi(r)=v_n(r/2)^n$ (where $v_n$ is the volume of the unit ball in $\re[n]$) satisfies all conditions required by the theorem. Notice that the given $\Phi(r)$ is the volume of one orthant of the ball of radius $r$ centered in $x$, and so to prove that such $\Phi(r)$ works we will show that for any point $x$ in the hypercube and a sufficiently small $r>0$, the ball of radius $r$ centered in $x$ will always contain at least one orthant strictly inside the hypercube.

    To do that, first define for any $x\in\SSpc$ its sign indicator as the vector $s\in\{-1,-1\}^n$ whose elements satisfy $s_i=1$ if $x_i\leq0.5$ and $s_i=-1$ otherwise. Then, define the closed cone $O_s:=\{u\in\re[n]~|~s_iu_i\geq0 ~i=1,\dots,n\}$ and the orthant of the ball $B_r(x)$ as $B^s_r(x):x+(B_r(0)\cap O_s)$. Clearly $B^s_r(x)\subset B_r(x)$, so we next show that for all $r< 1/2$ $B^s_r(x)\subset\SSpc$ -- and consequently $B^s_r(x)\subset\SSpc\cap B_r(x)$.

    To show that $B^s_r(x)\subset\SSpc$, take any $\overline{x}\in B^s_r(x)$. By definition $\overline x=x+u$ for some $u\in O_s$ satisfying $\|u\|\leq r$ -- consequently, for any given coordinate it also holds that $|u_i|\leq r$. For any given coordinate $i$, if $s_i>0$ then $u_i>0$, $\overline x_i=x_i+u_i\geq x_i\geq0$, and $\overline x_i=x_i+u_i\leq x_i+r\leq 0.5+0.5=1$. Alternatively if $s_i=-1$, then $u_i\leq 0$ which implies $\overline x_i=x_i+u_i\leq x_i\leq 1$ and $\overline x_i=x_i+u_i\geq x_i-r\geq 0$. Thus for all $i$, $0\leq\overline x_i\leq 1$ implying $\overline x\in\SSpc$ and thus $B^s_r(x)\subset\SSpc\cap B_r(x)$.

    Then, and since $\mu$ is the normalized Lebesgue measure in $\SSpc$, and is $0$ for all $x\not\in\SSpc$, it holds that
    \begin{equation*}
        \mu(B_r(x))\geq\mu(B_r^s(x))=v_n\left(\frac{r}{2}\right)^n=:\Phi(r)
    \end{equation*}
    for all $r\leq 1/2$.

    Next, define $\overline\alpha:=(1/c_3)(1+(\log(c_2/\delta)/\log(N))$ and pick $\alpha:=\max\left\{\alpha_0,\overline\alpha\right\}$ and notice that 
    \begin{equation*}
        c_2N^{1-c_3\alpha}\leq c_2N^{1-c_3\overline\alpha}=\delta,
    \end{equation*}
    and that
    \begin{align*}
        \alpha\log(N)&\geq\overline\alpha\log(N)\\
        &=\frac{1}{c_3}\left(1+\frac{\log(\frac{c_2}{\delta})}{\log{N}}\right)\log(N)\\
        &= c_3^{-1}\left(\log(N)+\log\left(\frac{c_2}{\delta}\right)\right)\\
        &= c_3^{-1}\log\left(\frac{c_2N}{\delta}\right)
    \end{align*}
    With the above choice of $\Phi$ and $\alpha$ and through Theorem 2.1 of \cite{reznikov2016covering} we have established that
    \begin{equation*}
        \mathbb{P}\left[h(\xsp,\SSpc)\geq c_1\Phi^{-1}\left(\frac{\alpha\log(N)}{N}\right)\right]\leq c_2N^{1-c_3\alpha}.
    \end{equation*}
    which is equivalent to
    \begin{equation*}
        \mathbb{P}\left[h(\xsp,\SSpc)\geq k_1\left(\frac{\log\left(k_2N/\delta\right)}{N}\right)^{1/n}\right]\leq\delta
    \end{equation*}
    with $k_1=c_1(c_3 v_n)^{-1/n}/2$ and $k_2=c_2$. Finally, we invert the inequalities to get
    \begin{equation*}
        \mathbb{P}\left[h(\xsp,\SSpc)\leq k_1\left(\frac{\log\left(k_2N/\delta\right)}{N}\right)^{1/n}\right]\geq1-\delta.
    \end{equation*}
    The remainder of the proof is simply applying Corollary \ref{cor:highprobbnd} to this high probability bound for the covering radius.
\end{proof}

Observe that Theorem \ref{thm:HighProbpUniDist} concludes that the generalization error of $\fmin$ is independent of the choice of $\fmin$, being a function only of the chosen $\rho,\eps,N$ and $\delta$. Moreover, this error can be made arbitrarily small as the sample size $N\to\infty$, constrained only by the noise level $\er$.

\subsection{Discussion and extension to general measures over $\SSpc$}

In this paper, we focus on uniform approximation of functions, and therefore uniform sampling over $\SSpc$ is a natural strategy for generating training data. If, instead, the objective is to bound the \emph{expected} error with respect to a probability measure $\mu$,
\begin{equation}
\Loss_\mu(f,g):=\int_{\SSpc}|f(x)-g(x)|\,\mathrm{d}\mu(x),
\end{equation}
then it is natural to use the same distribution for both training and testing. In this setting, the key observation in Lemma~\ref{lem:UBOrigiAnyPt} continues to hold for any absolutely continuous measure, as well as for measures supported on lower-dimensional manifolds, since the result is point-wise. Consequently, one obtains bounds on the expected error that depend on the expected distance between a test sample and its nearest training point.
If the training points $x_1,\dots,x_N$ can be selected freely, their optimal placement reduces essentially to a vector quantization problem \cite{gersho2012vector}.
Singular measures (i.e., those supported on lower-dimensional sets) are particularly relevant in light of Theorem~\ref{thm:HighProbpUniDist}, which shows that generalization error suffers from the curse of dimensionality. In many high-dimensional applications (e.g., images or videos), it is common to invoke the “manifold hypothesis,” according to which data lie near a lower-dimensional manifold embedded in the ambient space. We leave a detailed investigation of these extensions to future work.

\subsection{Illustrative Application: Stable Vector Field Estimation}

Consider a stable nonlinear system of the form
\begin{equation}\label{eq:true-dyn}
\dot{x}=g(x)
\end{equation}
where $x(t)\in \R^n$ denotes the state at time $t$, and $g$ is a sufficiently smooth vector field. The objective is to learn a vector field $f$ from a noise-free dataset $\DSet(g)$ over a compact domain $\mathcal{D}$ such that: 1) $f$ approximates $g$ uniformly on $\mathcal{D}$, and 2) the trajectory generated by $f$ is also uniformly close to the true trajectory. 

For the first objective, Theorem~\eqref{thm:ThrGuarantee2} shows that \eqref{eq:lipregprob} with $\eps \ge 0$ can yield a function $f \in L_0$ satisfying
\begin{equation}\label{eq:fit-err}
\|f-g\|_{\infty,\mathcal{D}} \le (2l_g+\rho) h + 2\eps
\end{equation}
where $\rho \ge 0$ and $h = \sup_{x\in\mathcal{D}} h(x,\DSet(g))$.

To address the second objective, we consider system \eqref{eq:true-dyn} subject to an additive perturbation $u(t)\in\R^n$, namely
\begin{equation}\label{eq:pert-dyn}
\dot{x}=g(x)+u.
\end{equation}
Let $x_g(\cdot; x_0, u)$ be the solution of \eqref{eq:pert-dyn} under input $u$ and initial condition $x_0$.
\begin{definition}
    System \eqref{eq:pert-dyn} is said to be incremental bounded input bounded output ($\delta$BIBO) stable if there exist a function $\gamma\in\mathcal{K}$ such that for any $x_0\in\R^n$ and bounded inputs $u^a,u^b$, it holds that 
\begin{equation}\label{eq:issdef}
    \begin{split}
        \|x_g\bigl(\cdot; x_0, u^a\bigr)-x_g\bigl(\cdot; x_0, u^b\bigr)\|_\infty \leq\gamma\bigl(\|u^a-u^b\|_{\infty}\bigr)
    \end{split}
\end{equation}
where $\|u\|_\infty := \sup_{t\geq 0} \|u(t)\|$.
\end{definition}


\begin{proposition}
    Suppose that \eqref{eq:pert-dyn} is $\delta$BIBO stable and $f$ is a solution of \eqref{eq:lipregprob} satisfying \eqref{eq:fit-err}. Let $\mathcal{D}_f\subset \mathcal{D}$ be a forward-invariant set of $\dot{x}=f(x)$. Then, for any $x_0\in\mathcal{D}_f$ we have
    \begin{equation}
        \label{eq:VecFieldEstimation}
        \left\|x_{f}(\cdot; x_0)-x_{g}(\cdot; x_0)\right\|_\infty\leq \delta
    \end{equation}
    where $x_f(\cdot; x_0)$ and $x_g(\cdot; x_0)$ are the solutions of vector fields $f$ and $g$, respectively. 
\end{proposition}
 
\begin{proof}
First, we rewrite the learned dynamics $\dot{x}=f(x)$ into the form \eqref{eq:pert-dyn} with $u=f-g$. Since $D_f$ is forward-invariant under the vector field $f$, we have $x(t)\in \mathcal{D}$, which further implies $\|u(x(t))\|\leq (2l_g+\rho) h + 2\eps$. Then, \eqref{eq:VecFieldEstimation} follows by taking $u^a=u$ and $u^b=0$ in \eqref{eq:issdef}.
\end{proof}

\section{Neural Network Implementation}

Given the dataset $\DSet(g)$, we define an empirical lower bound for the Lipschitz constant $l_g$ as
\begin{equation}\label{eq:lipdata}
L_{\mathrm{data}} := \max_{1\le i<j\le N}
\frac{\|g(x_i)-g(x_j)\|}{\|x_i-x_j\|},
\end{equation}
where $x_i \neq x_j$. By the Kirszbraun theorem, there exists a Lipschitz function $f^\star:\mathbb{R}^n\to\mathbb{R}^n$ satisfying
\[
\mathrm{Lip}(f^\star)=L_{\mathrm{data}} \; \text{ and }\;
f^\star(x_i)=g(x_i), \ \forall i=1,\ldots,N .
\]
Hence $f^\star$ is an optimal solution to Problem~\eqref{eq:lipregprob} with $\eps=\overline{\eps}=0$. An explicit construction of $f^\star$ is provided in \cite{azagra2021kirszbraun}; however, evaluating $f^\star(x)$ requires computing the gradient of the convex hull of an infimum over the data points, which is computationally demanding.

A natural approach is to learn an approximate solution using a multilayer perceptron (MLP) $y=f_\theta(x)$ defined as
\begin{equation}\label{eq:MLP}
\begin{aligned}
h_0 &= x,\\
h_k &= f_k(h_{k-1}) := \phi(W_k h_{k-1}+b_k),\\
y &= W_{K+1}h_K + b_{K+1},
\end{aligned}
\end{equation}
for $k=1,\ldots,K$, where $h_k\in\R^{n_k}$ denotes the hidden unit, $W_k\in\R^{n_k\times n_{k-1}}$ and $b_k\in\R^{n_k}$ are learnable parameters, and the activation $\phi$ is a monotone and $1$-Lipschitz scalar function applied elementwise. The vector $\theta\in\R^p$ collects all learnable parameters. When $p$ is sufficiently large, $f_\theta$ can interpolate the dataset $\DSet(g)$. However, $f_\theta$ cannot be directly applied to solve  \eqref{eq:lipregprob} since computing $\mathrm{Lip}(f_\theta)$ is NP-hard in general \cite{virmaux2018lipschitz}. To address this issue, we consider Lipschitz-bounded neural networks of the form $y=f_{\theta,\eta}(x)$, referred to as \emph{LipNet}:
\begin{equation}
\begin{aligned}
h_0 &= \sqrt{\eta}\,x,\\
h_k &= \bar{f}_k(h_{k-1}), \quad k=1,\ldots,K,\\
y &= \sqrt{\eta}\,h_K,
\end{aligned}
\end{equation}
where each $\bar{f}_k$ is a $1$-Lipschitz layer. By construction,
\[
\mathrm{Lip}(f_{\theta,\eta})
\le \eta \prod_{k=1}^K \mathrm{Lip}(\bar f_k)
= \eta .
\]
During training, it is convenient to reparameterize $\eta$ as $\eta=e^\psi$ with $\psi\in\R$ as a free parameter.

Various parameterizations for $1$-Lipschitz layers have been proposed; see the recent survey \cite{manchester2026neural} for an overview. One simple example is the spectral parameterization \cite{miyato2018spectral}
\begin{equation}\label{eq:spectral}
\bar f(x) = \phi\!\left(\frac{1}{\bar\sigma(W)}Wx + b\right),
\end{equation}
where $\bar\sigma(W)$ denotes the maximum singular value of $W$. For any $W$ and $b$, this layer has a certified Lipschitz bound of $1$, although the bound may be conservative. Other constructions, such as orthogonal layers \cite{trockman2021orthogonalizing}, SLL layers \cite{araujo2023unified}, and Sandwich layers \cite{wang2023direct}, can yield tighter bounds. We construct optimization formulations of Problem~\eqref{eq:lipregprob} with $\eps=0$ (P1) two proxy formulations with an approximation hyperparameter $\rho$ (P2 and P3):

\begin{itemize}
\item[P1)] \textbf{Lipschitz-bound minimization}
\begin{equation}
\begin{aligned}
\underset{\theta,\eta}{\text{minimize}} \quad &
\max(\eta,L_{\mathrm{data}}) \\
\text{subject to}\quad &
f_{\theta,\eta}(x_i)=g(x_i), \quad i=1,\ldots,N .
\end{aligned}
\end{equation}

\item[P2)] \textbf{Feasibility formulation}
\begin{equation}
\begin{aligned}
\underset{\theta,\eta}{\text{minimize}} \quad & 0 \\
\text{subject to}\quad
& \eta-L_{\mathrm{data}}-\rho \le 0, \\
& f_{\theta,\eta}(x_i)=g(x_i), \quad i=1,\ldots,N ,
\end{aligned}
\end{equation}
where $\rho\ge0$ is a hyperparameter.

\item[P3)] \textbf{Lipschitz-constrained supervised learning}
\begin{equation}
\underset{\theta}{\text{minimize}} \quad
\frac{1}{N}\sum_{i=1}^N
\|f_{\theta,\eta}(x_i)-g(x_i)\|^2 ,
\end{equation}
where $\eta=\rho+L_{\mathrm{data}}$ with  hyperparameter $\rho\ge0$. 
\end{itemize}

Problems P1 and P2 are solved using the augmented Lagrangian method  \cite{bertsekas2014constrained}. The benefit of Problem P3 is that it is a standard supervised learning problem, and can besolved using standard stochastic gradient descent or its variants.

\subsection{Numerical Example}


We consider the following stable nonlinear dynamics
\begin{equation}\label{eq:true-model}
    \dot{x}=g(x):=\begin{bmatrix}
        -x_1 + x_3 \\
        \frac{1}{2}x_1^2-x_1x_3+x_3-x_2 \\
        -x_1 - x_3 
    \end{bmatrix}.
\end{equation}
The training dataset $\DSet(g)$ consists of $27$ grid samples (three points along each state dimension) together with $500$ samples drawn uniformly from the domain $\mathcal{D}=[-2,2]\times[-10,10]\times[-2,2]$. The corresponding empirical lower bound in \eqref{eq:lipdata} is $L_{\mathrm{data}}=4.3282$. The test dataset contains $10{,}000$ samples drawn uniformly from $\mathcal{D}$.

We train the proposed LipNet using Sandwich layers \cite{wang2023direct}. For comparison, we also train a standard multilayer perceptron (MLP) using supervised learning. Both networks consist of six hidden layers with width $128$. The following performance metrics are reported:
\[
\begin{aligned}
\mathrm{MSE} &= \frac{1}{N}\sum_{i=1}^{N}\|f(x_i)-g(x_i)\|^2,\\
\mathrm{Max} &= \max_{1\le i\le N} \|f(x_i)-g(x_i)\|,
\end{aligned}
\]
Both metrics are evaluated on the training and test datasets.

\begin{table}[!tb]
    \centering
    \scalebox{0.8}{
    \begin{tabular}{c|c|cc|cc|cc}
    \toprule
    \midrule
    \multicolumn{2}{c|}{} & \multicolumn{2}{c|}{Training} & \multicolumn{2}{c|}{Test} & \multicolumn{2}{c}{Lipschitz bounds} \\
    \midrule
    Model & H.P. $\rho$ & MSE & Max & MSE & Max & Emp. $\underline{L}$ & Cert. $\overline{L}$\\
    \midrule
    \midrule
    MLP & -- & 1.0e-9\phantom{0} & 9.8e-4 & 2.8e-1 & 2.95 & 76.1062 & -- \\
    \midrule
    LipNet-P1 & -- & 1.0e-10 & 6.7e-5 & 1.2e-3 & 0.34 & \phantom{0}4.6021 & 5.3692\\
    \midrule
    \multirow{3}{*}{LipNet-P2} & $0.1 L_{\text{data}}$ & 3.1e-10 & 2.0e-4 & 1.4e-3 & 0.42 & \phantom{0}4.5207 & 4.7608\\
     & $0.05 L_{\text{data}}$ & 3.4e-10 & 3.2e-4 & 1.9e-3 & 0.46 & \phantom{0}4.4233 & 4.5446\\
     & $0$ & 9.3e-9\phantom{0} & 2.5e-3 & 3.6e-3 & 0.54 & \phantom{0}4.2966 & 4.3282\\
    \midrule
    \multirow{3}{*}{LipNet-P3} & $0.1L_{\text{data}}$ & 4.9e-10 & 5.0e-4 & 2.2e-3 & 0.54 & \phantom{0}4.4519 & 4.7610\\
     & $0.05L_{\text{data}}$ & 6.4e-10 & 6.8e-4 & 2.4e-3 & 0.54 & \phantom{0}4.4111 & 4.5446\\
     & $0$ & 4.3e-7\phantom{0} & 1.4e-2 & 3.3e-3 & 0.61 & \phantom{0}4.3000 & 4.3282\\
    \bottomrule
    \end{tabular}
    }
    \caption{Results for MLP and LipNet with $L_{\mathrm{data}}=4.3282$.}\label{tab:train}
\end{table}

\begin{figure}[!tb]
    \centering
    \includegraphics[width=0.9\linewidth]{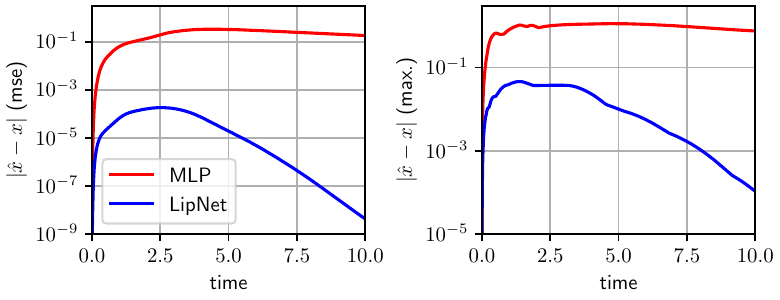}
    \caption{Trajectory errors of different models over 500 test data samples.}
    \label{fig:simu-err}
\end{figure}

\begin{figure}[!tb]
    \centering
    \includegraphics[width=\linewidth]{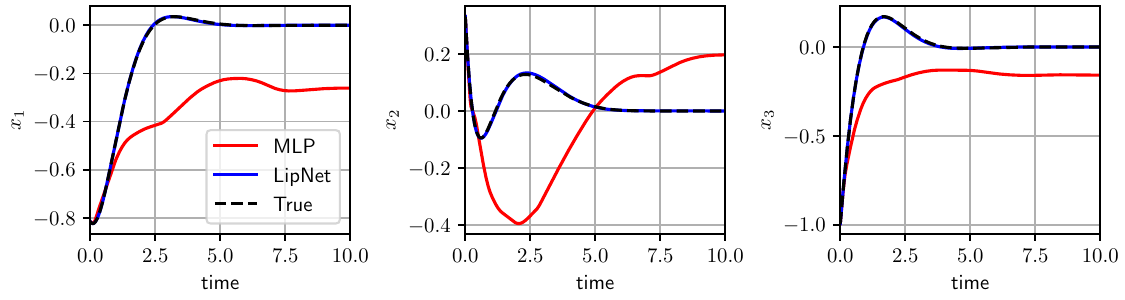}
    \caption{A trajectory sample of learnt models.}
    \label{fig:simu-traj}
\end{figure}

Table~\ref{tab:train} summarizes the training and test performance. Both MLP and LipNet achieve small training errors, indicating that the dataset can be interpolated by both models. However, the test error of the MLP is significantly larger. This behavior is consistent with its large Lipschitz bound, which leads to poor generalization of the learned vector field.

To evaluate the learned dynamics, we simulate the models for $10$ seconds using $500$ test samples as initial conditions. Figure~\ref{fig:simu-err} shows that the MSE simulation error of LipNet remains below $10^{-3}$ and decreases as the trajectories approach the equilibrium. In contrast, the MSE simulation error of MLP quickly increases to  $0.1$ and remains above this level. Figure~\ref{fig:simu-traj} illustrates a representative trajectory. The MLP model converges to a different equilibrium point, indicating that the learned vector field does not preserve the qualitative behavior of the true system.

Figure~\ref{fig:fit} compares the learned vector fields for different values of $x_2$. LipNet preserves the smooth structure across different $x_2$, whereas the MLP produces a non-smooth field whose structure varies with $x_2$.

Table~\ref{tab:mlp-weight-decay} investigates the effect of weight decay, a widely applied regularization technique. Increasing the weight decay coefficient reduces the empirical Lipschitz bound and improves test accuracy. However, this does not yield certified Lipschitz bounds and comes at the cost of significantly worse interpolation performance. Furthermore, choosing the right regularization hyperparameter requires a separate validation data set, whereas Lipschitz-minimal interpolation can be applied with access only to training data.

Among the proposed formulations, the Lipschitz bound minimization approach (P1) achieves the best overall performance. Considering different $1$-Lipschitz layers in Table~\ref{tab:layer}, Sandwich layers produce tightest Lipschitz bounds and best approximation accuracy, closely followed by Orthogonal layers, while performance was significantly worse with the spectral layer as expected.

\begin{figure*}[!tb]
    \centering
    \begin{tabular}{c}
        \includegraphics[width=0.82\linewidth]{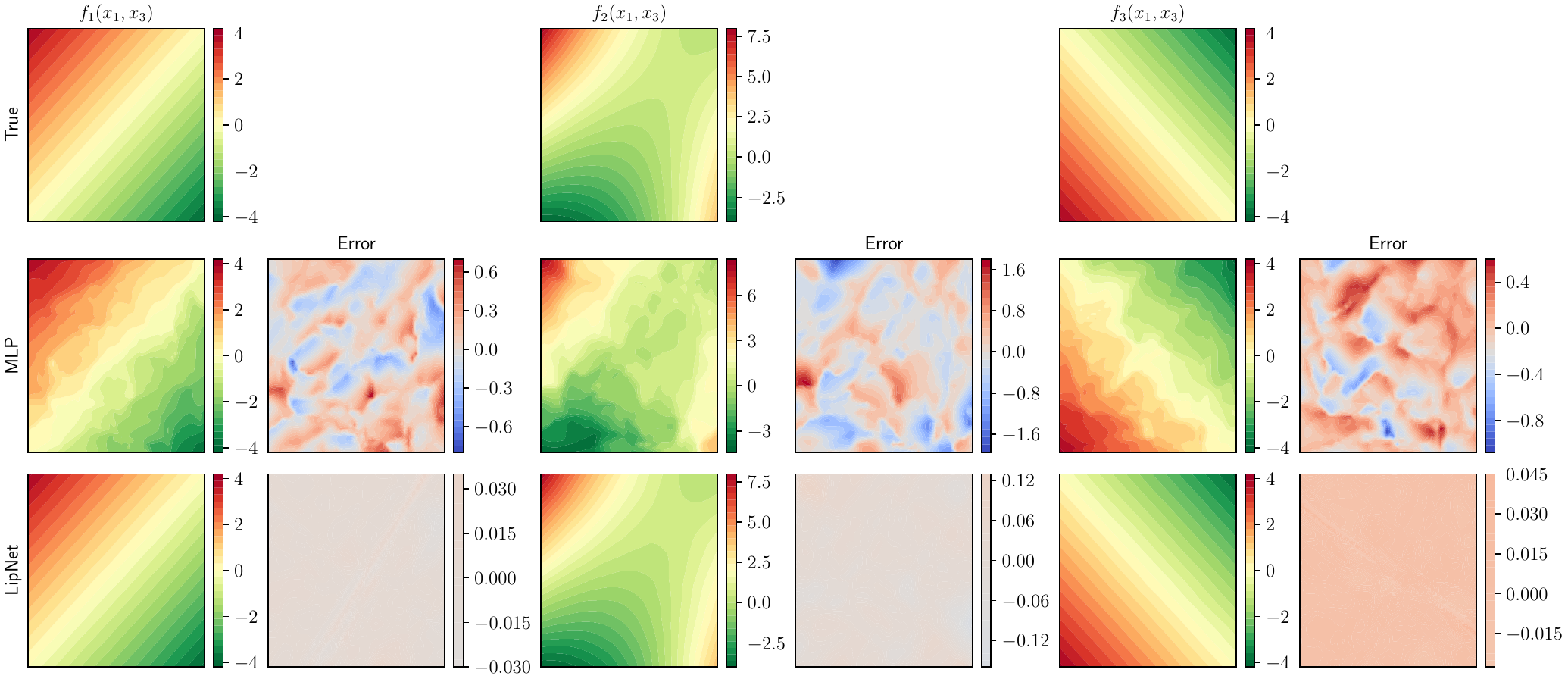} \\
        (a) $x_2=0$ \\
        \includegraphics[width=0.82\linewidth]{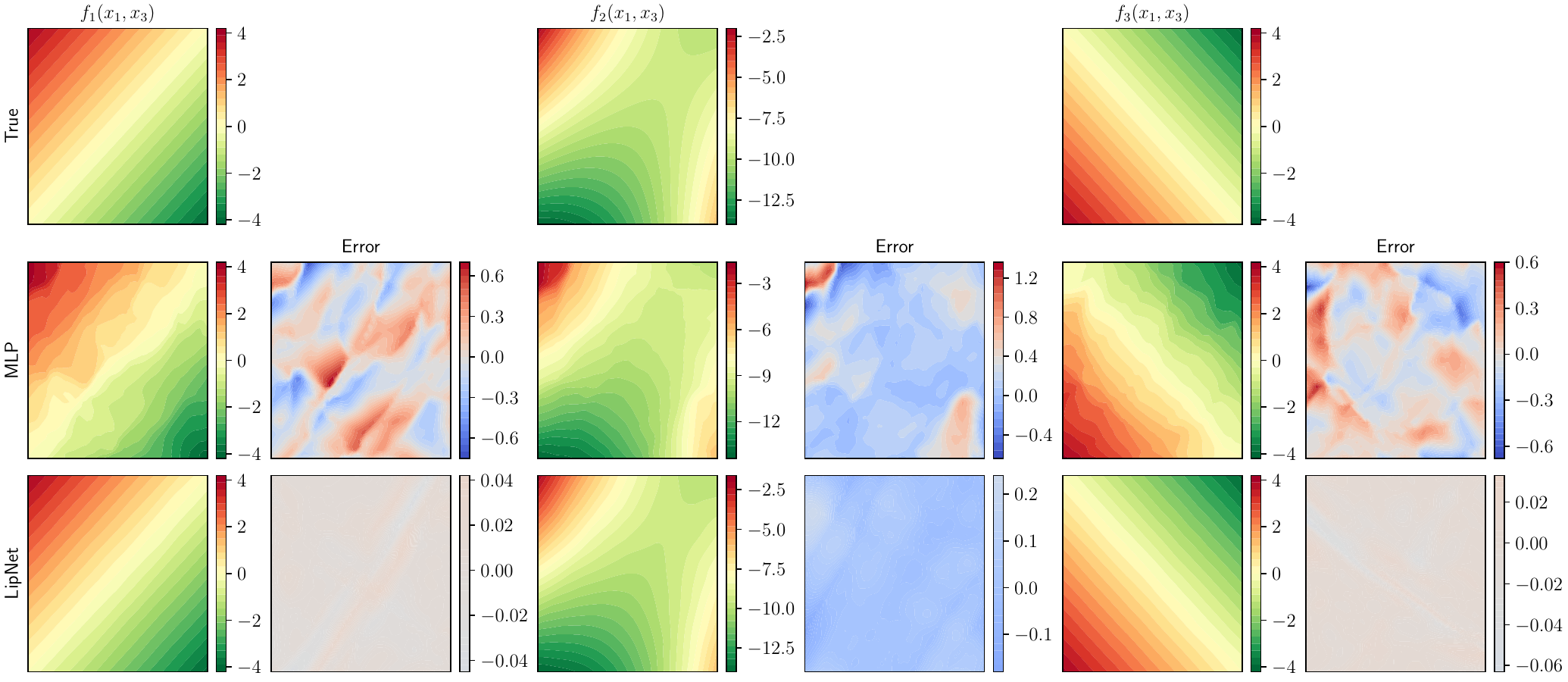} \\
        (b) $x_2=10$
    \end{tabular}
    \caption{The output channel $f_i$ of the learnt models with fixed $x_2$ over the region $(x_1, x_3) \in [-2,2]\times[-2,2]$.}
    \label{fig:fit}
\end{figure*}

\begin{table}[!tb]
    \centering
    \scalebox{0.9}{
    \begin{tabular}{c|cc|cc|c}
    \toprule
    \midrule
    MLP & \multicolumn{2}{c|}{Training} & \multicolumn{2}{c|}{Test} & Lip. bounds \\
    \midrule
    \midrule
    weight decay & MSE & Max. & MSE & Max. & Emp. $\underline{L}$ \\
    \midrule
    0.001 & 9.4e-10 & 8.6e-4 & 2.1e-1 & 2.72 & 54.3703 \\
    0.005 & 1.2e-8 & 3.2e-3 & 1.0e-1 & 1.69 & 24.3221 \\
    0.01\phantom{0}  & 1.5e-8 & 3.8e-3 & 8.2e-2 & 2.13 & 17.4451 \\
    0.05\phantom{0}  & 3.2e-7 & 1.0e-2 & 2.4e-2 & 1.19 & 11.1999 \\
    0.1\phantom{00}  & 1.0e-5 & 3.5e-2 & 1.1e-2 & 0.72 & \phantom{0}7.5957 \\
    0.5\phantom{00}  & 6.6e-5 & 5.6e-2 & 3.9e-3 & 0.38 & \phantom{0}5.6107 \\
    1.0\phantom{00}  & 2.5e-4 & 1.4e-1 & 4.2e-3 & 0.36 &  \phantom{0}4.8761 \\
    5.0\phantom{00}  & 7.6e-1 & 4.26 & 2.42 & 4.33 &  11.4525 \\
    \bottomrule
    \end{tabular}
    }
    \caption{Results for MLP with different weight decays.}\label{tab:mlp-weight-decay}
\end{table}

\begin{table}[!tb]
    \centering
    \scalebox{0.76}{
    \begin{tabular}{c|c|cc|cc|cc}
    \toprule
    \midrule
    \multicolumn{2}{c|}{LipNet} & \multicolumn{2}{c|}{Training} & \multicolumn{2}{c|}{Test} & \multicolumn{2}{c}{Lipschitz bounds} \\
    \midrule
    setup & layer & MSE & Max. & MSE & Max. & Emp. $\underline{L}$ & Cert. $\overline{L}$\\
    \midrule
    \midrule
    \multirow{3}{*}{P1}  & Sandwich & 1.0e-10 & 6.7e-5 & 1.2e-3 & 0.34 & 4.6021 & 5.3692\\
    & Orthogonal & 3.7e-10 & 3.1e-4 & 1.3e-3 & 0.37 & 4.5591 & 5.7713\\
     & Spectral & 4.9e-8\phantom{0} & 1.7e-3 & 7.2e-4 & 0.29 & 5.0955 & 8.0437\\
    \midrule
    \multirow{3}{*}{P2} & Sandwich & 3.1e-10 & 2.0e-4 & 1.4e-3 & 0.42 & 4.5207 & 4.7608\\
     & Orthogonal & 4.8e-8\phantom{0} & 5.1e-4 & 2.9e-3 & 0.47 & 4.3528 & 4.7610\\
     & Spectral & 8.9e-5\phantom{0} & 9.0e-2 & 2.1e-3 & 0.42 & 4.3554 & 4.7614\\
    \midrule
    \multirow{3}{*}{P3} & Sandwich & 4.9e-10 & 5.0e-4 & 2.2e-3 & 0.54 & 4.4519 & 4.7610\\
     & Orthogonal & 9.1e-6\phantom{0} & 5.8e-2 & 3.1e-3 & 0.49 & 4.3251 & 4.7610\\
     & Spectral & 1.0e-4\phantom{0} & 1.8e-1 & 2.1e-3 & 0.41 & 4.2927 & 4.7610\\
    \bottomrule
    \end{tabular}
    }
    \caption{Results for LipNet with different 1-Lipschitz layers. We use $\rho=0.1L_{\text{data}}$ for P2 and P3. }\label{tab:layer}
\end{table}

\section{Conclusions}

In this work, we presented a theoretical framework for Lipschitz-minimal interpolation and established explicit generalization guarantees for functions that approximately interpolate finite datasets while minimizing their Lipschitz constant. The results show that such interpolants admit uniform error bounds that depend linearly on the sampling density through the covering radius, and the noise level. Importantly, these guarantees are independent of the specific choice of interpolating function and instead depend only on structural properties of the problem, highlighting the role of smoothness as an implicit regularizer in overparameterized models.

We further showed that these guarantees extend to practical model classes, such as neural networks, under mild density assumptions in the 
$C^{0,1}$ topology. The proposed Lipschitz-constrained formulations can be implemented using Lipschitz-bounded architectures and standard optimization techniques, enabling tractable approximations of Lipschitz-minimal solutions. Numerical experiments demonstrate that enforcing Lipschitz bounds significantly improves generalization and preserves qualitative system behavior, particularly in dynamical systems settings.

Several directions for future work remain. These include extending the analysis to non-uniform sampling distributions and data supported on lower-dimensional manifolds, improving computational methods for Lipschitz-constrained optimization, and tightening the dependence of the bounds on dimension. More broadly, the framework suggests a principled approach to incorporating smoothness priors in learning systems, with potential applications in robust control, system identification, and safety-critical machine learning.


\bibliographystyle{ieeetran}
\bibliography{references}
\extendedversion{
\begin{appendix}
    \subsection{Alternate proof for Theorem \ref{thm:ThrGuarantee2}(i)}
    In this appendix we will present an alternative proof of claim (i) of Theorem \ref{thm:ThrGuarantee2}. As a reminder, the claim is as follows:
    \begin{claim}
        For any $\rho>0$, $\fgen\in\LSpace$, and $\nDSet(\fgen)$, let $\FSet(\SSpc)\subset C^1(\SSpc)$ be a set of functions dense in the $C^{0,1}(\SSpc)$ topology, and $\eps\geq\er$ be the tolerance of \eqref{eq:lipregprob}. Further, let $\fopt$ be any solution of \eqref{eq:lipregprob}. Then there always exist an $\fmin\in\FSet$ that is both a feasible solution of \eqref{eq:lipregprob} and that satisfies $\lc{\fmin}\leq\lc{\fopt}+\rho$.
    \end{claim}

    To begin proving the claim, define the McShane extension $\tilde\fgen:\re[n]\to\re$, of any $\fgen\in \LSpace$, as
    \begin{equation*}
        \tilde \fgen(x) := \inf_{y\in\SSpc}(\fgen(y)+\lc{\fgen}\|x-y\|)
    \end{equation*}
    and its Moreau-Yosida inf-convolution $\fgen_\lambda$ as
    \begin{equation*}
        \fgen_\lambda:=\inf_{y\in\re[n]}\left(\tilde{\fgen}(y)+\frac{1}{2\lambda}\|x-y\|^2\right)
    \end{equation*}
    for any $\lambda\in\re$. Notice that by construction $\tilde{\fgen}$ is Lipschitz with $\lc{\tilde{\fgen}}\leq \lc{\fgen}$ and $\fgen_\lambda$ is finite and $C^{1,1}(\re[n])$ (see, \cite{BauschkeCombettes,LasryLions}). With these definitions, the proof will go as follows: we will build a sequence of smooth functions that converge uniformly to any given Lipschitz continuous function, using the Moreau-Yosida convolution; then we will show that the Lipschitz constants of the functions in such a sequence converge to the Lipschitz constant of the function that they are approximating; and finally we will use the fact that every element in the sequence can be approximated by its own sequence of functions in $\FSet$, due to assumed density of $\FSet$ in $C^{0,1}(\SSpc)$, thus completing the proof.
    
    The first step towards proving (i) is, then, to show that $\lim_{\lambda\downarrow0}\fgen_{\lambda} = \tilde{\fgen}$. To see that, notice that by Lipschitzness
    \begin{align*}
        &\tilde{\fgen}(y)\geq \tilde{\fgen}(x)-\lc{\fgen}\|y-x\|\\
        &\tilde{\fgen}(y)+\frac{1}{2\lambda}\|x-y\|^2\geq \tilde{\fgen}(x)+\left(\frac{1}{2\lambda}\|x-y\|^2-\lc{\fgen}\|x-y\|\right)\\
        &\tilde{\fgen}(y)+\frac{1}{2\lambda}\|x-y\|^2\geq \tilde{\fgen}(x)-\frac{\lambda \lc{\fgen}^2}{2}
    \end{align*}
    since $\min_{r>0}(1/2\lambda)r^2-\lc{\fgen}r=-(\lambda \lc{\fgen}^2)/2$. Taking the infimum over $y$ of the inequality above results in
    \begin{align*}
        &\fgen_{\lambda}(x)\geq \tilde{\fgen}(x)-\frac{\lambda \lc{\fgen}^2}{2}\\
        &\tilde{\fgen}(x)-\fgen_{\lambda}(x)\leq \frac{\lambda \lc{\fgen}^2}{2}
    \end{align*}
    proving that $\lim_{\lambda\downarrow0}\fgen_{\lambda} = \tilde{\fgen}$.

    Next we will show that over the compact set $\SSpc$, $\lim_{\lambda\downarrow 0}\lc{\fgen_{\lambda}}=\lc{\fgen}$. Since the Lipschitz semi-norm is lower semi-continuous in the uniform topology, we have that
    \begin{equation*}
        \lc{\fgen}\leq\liminf_{\lambda\downarrow0}\lc{\fgen_{\lambda}},
    \end{equation*}
    thus we need to prove the opposite inequality -- in particular we will next show that in $\SSpc$ it holds that $\lc{\fgen_{\lambda}}\leq \lc{\fgen}$ for all $\lambda>0$ -- i.e. that $\fgen_{\lambda}$ is $\lc{\fgen}$-Lipschitz. To do that, pick any two $x_1,x_2\in\re[n]$. Define $d:=x_2-x_1$ and for any $y\in\re[n]$, $z:=y+d$. Notice that $x_2-z=x_1-y$; we will leverage this inequality to prove the claim. Write
    \begin{align*}
        \fgen_\lambda(x_2)&\leq \tilde\fgen(z)+\frac{1}{2\lambda}\|x_2-z\|^2\\
        &=\tilde\fgen(y+d)+\frac{1}{2\lambda}\|x_1-y\|^2\\
        &\leq \tilde\fgen(y)+\frac{1}{2\lambda}\|x_1-y\|^2+\lc{\fgen}\|d\|,
    \end{align*}
    where the last inequality follows from the fact that $\tilde\fgen$ is $\lc{\fgen}$-Lipschitz. Taking the infimum of the expression above over $y$ gives
    \begin{equation*}
        \fgen_\lambda(x_2)\leq \fgen_\lambda(x_1)+\lc{\fgen}\|x_2-x_1\|,
    \end{equation*}
    and by symmetry, the same inequality follows with $x_1$ and $x_2$ reversed. Therefore,
    \begin{equation*}
        |\fgen_\lambda(x_2)-\fgen_\lambda(x_1)|\leq \lc{\fgen}\|x_2-x_1\|,
    \end{equation*}
    proving the claim. We have, therefore, established that 
    \begin{equation*}
        \lc{\fgen}\leq\liminf_{\lambda\downarrow0}\lc{g_\lambda}
    \end{equation*}
    and that
    \begin{equation*}
        \lc{g_\lambda}\leq\lc{\fgen},
    \end{equation*}
    for all $\lambda> 0$. Therefore, by taking the limit supremum of the inequality above with $\lambda\downarrow 0$ we show that the limit exists and converges fo $\lc{\fgen}$. 

    We are now ready to prove the main result. First pick any sequence $\lambda_m\downarrow0$ and for simplicity of notation define $h_m:=g_{\lambda_m}$. So far we have established that $\|h_m-g\|_{\infty}\to 0$ and that $\lc{h_m}\to\lc{g}$. We will then show that for any $m$, $h_m$ can be similarly approximated by a sequence of functions in $\FSet$.

    Since $h_m$ is smooth, and since $\FSet$ is dense in the $C^{0,1}$ topology, then for any sequence $\delta_m\downarrow 0$ there exists an $f_m\in\FSet$ such that 
    \begin{equation*}
        \|f_m-h_m\|_{\infty,\SSpc}+\|f'_m-h'_m\|_{\infty,\SSpc}\leq\delta_m
    \end{equation*}
    which implies in particular that $\|f_m-h_m\|_{\infty,\SSpc}\leq \delta_m$ and $\|f'_m-h'_m\|_{\infty,\SSpc}\leq\delta_m$. From here we can show convergence of the sequence $f_m$ to $\fgen$ as
    \begin{align*}
        \|f_m-\fgen\|_{\infty,\SSpc}&\leq\|f_m-h_m\|_{\infty,\SSpc}+\|h_m-\fgen\|_{\infty,\SSpc}\\&\leq\delta_m+\|h_m-g\|_{\infty,\SSpc}
    \end{align*}
    which after taking a limit for $m\to\infty$ results in $\|f_m-\fgen\|_{\infty,\SSpc}= 0$. For convergence of the Lipschitz constant notice that
    %
    %
    \begin{equation*}
        |\|f_m'\|_{\infty,\SSpc}-\|h'_m\|_{\infty,\SSpc}\| |\leq\|f'_m-h_m'\|_{\infty,\SSpc}\leq\delta_m,
    \end{equation*}
    and, thus
    \begin{equation*}
        \lc{h_m}-\delta_m\leq\lc{f_m}\leq\lc{h_m}+\delta_m,
    \end{equation*}
    from denseness in the $C^{0,1}$ norm, which similarly results in $\lc{f_m}=\lc{\fgen}$ after taking the limit for $m\to\infty$. Finally, if $\SSpc$ is not convex, apply the same argument for the convex hull of $\SSpc$, which allows for the same conclusion to follow (since the convex hull of $\SSpc$ contains $\SSpc$). At this point, the proof is complete.\hfill $\blacksquare$

\end{appendix}}

\end{document}

%% file: MathDefs.tex
\DeclareMathOperator*{\argmin}{arg\,min}

\newcommand{\na}{\mathbb{N}}
\newcommand{\ir}{\mathbb{Z}}

\newcommand{\re}[1][\empty]{\mathbb{R}^{#1}}
\newcommand{\rep}[1][\empty]{\re[#1]_{+}}
\newcommand{\repp}[1][\empty]{\re[#1]_{++}}
\newcommand{\nat}{\mathbb{N}}

\newcommand{\dist}{\mathbf{d}}

\newcommand{\er}{\overline\varepsilon}

\newcommand{\DSet}[1][\empty]{\ifx#1\empty \mathbb{D}_{\NSet} \else \mathbb{D}_{\NSet}(#1)\fi}
\newcommand{\nDSet}[1][\empty]{\ifx#1\empty \mathbb{D}_{\NSet}^{\er} \else \mathbb{D}_{\NSet}^{\er}(#1)\fi}
\newcommand{\NSet}{N}
\newcommand{\xsp}[1][\empty]{\ifx#1\empty \mathbf{x}\else x_{#1}\fi}
\newcommand{\ysp}[1][\empty]{\ifx#1\empty \mathbf{y}\else y_{#1}\fi}

\newcommand{\SSpc}{\mathcal{D}}

\newcommand{\SDim}{n}

\newcommand{\LSpace}{L(\SSpc)}
\newcommand{\lc}[1]{l_{#1}}

\newcommand{\loss}[1][\empty]{\ifx#1\empty\ell\else \ell(#1)\fi}
\newcommand{\Loss}{\mathcal{L}}
\newcommand{\Lip}[1][\empty]{\ifx#1\empty \mbox{Lip}\else\mbox{Lip}(#1)\fi}

\newcommand{\lopt}{\lc{\fopt}}

\newcommand{\func}[1][\empty]{\ifx#1\empty f \else f(#1)\fi}
\newcommand{\funcg}[1][\empty]{\ifx#1\empty g \else f(#1)\fi}
\newcommand{\fgen}[1][\empty]{\ifx#1\empty g \else g(#1)\fi}
\newcommand{\fopt}[1][\empty]{\ifx#1\empty f^* \else f^*(#1)\fi}
\newcommand{\fmin}[1][\empty]{\ifx#1\empty \underline{f} \else \underline{f}(#1)\fi}
\newcommand{\FSet}{L_0}

\newcommand{\eps}{\varepsilon}

\newcommand{\PPt}{\mathbf{W}}
\newcommand{\POpt}{\PPt^*}

\newcommand{\NNet}[1][\empty]{\ifx#1\empty f_{\PPt} \else f_{\PPt}(#1)\fi}
\newcommand{\NOpt}[1][\empty]{\ifx#1\empty f_{\POpt} \else f_{\POpt}(#1)\fi}